\documentclass[a4paper, 11pt]{article} 
\RequirePackage{fullpage}

\usepackage[usenames,dvipsnames]{xcolor}
\usepackage[colorlinks=true,linkcolor=black,citecolor=BlueViolet,urlcolor=Maroon]{hyperref}
\usepackage{subcaption}
\usepackage{charter}
\usepackage{amsthm, amsmath, amssymb}
\usepackage[euler-digits,small]{eulervm} 
\usepackage{enumerate} 

\usepackage{tikz}
\usetikzlibrary{quotes}
\tikzstyle{filled vertex}=[fill=cyan, circle, thin, draw, inner sep=2pt, minimum width=3pt]
\tikzstyle{empty vertex}  = [{circle, draw, fill = white, inner sep=1.5pt}]
\tikzstyle{optional vertex}  = [{circle, densely dotted, draw=gray, fill = white, inner sep=1.5pt}]
\tikzstyle{normal edge}=[thick]

\newtheorem{theorem}{Theorem}[section]
\newtheorem{proposition}[theorem]{Proposition}
\newtheorem{lemma}[theorem]{Lemma}
\newcommand{\myref}{\hyperref[p1]{($\star$)}}

\title{On Fork-free T-perfect Graphs}
\author{Yixin Cao\thanks{Department of Computing, Hong Kong Polytechnic University, Hong Kong, China.  {\tt yixin.cao@polyu.edu.hk, shenghua.wang@connect.polyu.hk}.} 
  \and Shenghua Wang\footnotemark[1]
}
\date{}

\begin{document}

\maketitle

\begin{abstract}
  In an attempt to understanding the complexity of the independent set problem, Chv{\'a}tal defined t-perfect graphs.  While a full characterization of this class is still at large, progress has been achieved for claw-free graphs [Bruhn and Stein, Math.\ Program.\ 2012] and $P_{5}$-free graphs [Bruhn and Fuchs, SIAM J.\ Discrete Math.\ 2017].  We take one more step to characterize fork-free t-perfect graphs, and show that they are strongly t-perfect and three-colorable.  We also present polynomial-time algorithms for recognizing and coloring these graphs.
\end{abstract}

\section{Introduction}\label{sec:intro}

When Berge~\cite{berge-60-perfect-conjecture} introduced perfect graphs in the 1960s, the concept was purely graph-theoretic.  One of his conjectures, now known as the weak perfect graph theorem, was resolved with the help of polyhedral combinatorics~\cite{fulkerson-72-anti-blocking-polyhedra, lovasz-72-perfect-graphs}.  Padberg~\cite{padberg-74-perfect-matrices} and Chv{\'{a}}tal~\cite{chvatal-75-graph-polytopes} independently characterized the independent set polytope of a perfect graph, i.e., the convex hull of the characteristic vectors of all its independent sets, by showing that this polytope can be described by nonnegativity and clique constraints.
The strong perfect graph theorem~\cite{chudnovsky-06-strong-perfect-graphs-theorem}, another conjecture from the original paper of Berge on perfect graphs, asserts that a graph $G$ is perfect if and only if $G$ does not contain any odd hole (an induced odd cycle of length at least five) or the complement of an odd hole.
Naturally, we would like to have a similar characterization for t-perfect graphs.
Unfortunately, the progress toward this goal has been embarrassingly slow: even a working conjecture on it is still missing.

Both perfect graphs and t-perfect graphs were formulated to understand the complexity of the independent set problem, i.e., finding an independent set of the maximum weight.  Indeed, polynomial-time algorithms for the maximum independent set problem are known for perfect graphs~\cite{grotschel-86-relaxations-vertex-packing} and t-perfect graphs~\cite{eisenbrand-03-independent-set-t-perfect}.
There is another line of research on the independent set problem.  Algorithms for maximum matching can be used to solve the independent set problem on line graphs.  The idea can be generalized to quasi-line graphs and then to claw-free graphs~\cite{sbihi-80-independent-claw-free, minty-80-independent-claw-free}.  The independent set problem is simple on $P_3$-free and $P_4$-free graphs.  These results motivated the search of graphs $H$ such that the independent set problem can be solved in polynomial time when the input graph is $H$-free.   Note that there can be only a few such graphs $H$~\cite{alekseev-82-independent-set}.
Two subsequent additions are the fork~\cite{alekseev-04-classes-independent-set, Lozin08} and the $P_5$~\cite{lokshtanov-14-independent-p5-free}.  While the $P_5$ is a natural generalization of the $P_4$, the fork graph (also called the chair graph) is a generalization of both the claw and the $P_4$ (Figure~\ref{fig:claw-and-fork}).
We get a fork by attaching a private neighbor to a degree-one vertex of a claw, or a degree-two vertex of a $P_4$.

\begin{figure}[h]
  \centering
  \begin{subfigure}[b]{0.25\linewidth}
    \centering
    \begin{tikzpicture}
      \foreach \i in {1,..., 3} {
        \node[empty vertex] (v\i) at ({270 - \i * (360 / 4)}:1) {};
      }
      \node[empty vertex] (c) at (0:0) {};
      \foreach \i in {1,..., 3} {
        \draw (c) -- (v\i);
      }
    \end{tikzpicture}
    \caption{}
  \end{subfigure}
  \;
  \begin{subfigure}[b]{0.25\linewidth}
    \centering
    \begin{tikzpicture}[scale=1.]
      \foreach \i in {1,..., 3} {
        \node[empty vertex] (v\i) at (\i, 0) {};
      }
      \node[empty vertex] (v4) at (0, -0.5) {};
      \node[empty vertex] (v5) at (0, 0.5) {};
      \draw (v3) -- (v2) -- (v1) -- (v4) (v1) -- (v5);
    \end{tikzpicture}
    \caption{}
  \end{subfigure}
  \caption{(a) The claw graph and (b) the fork graph.}
  \label{fig:claw-and-fork}
\end{figure}
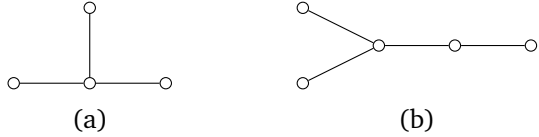

Since we know when a claw-free graph~\cite{Bruhn12} or a $P_{5}$-free graph~\cite{Bruhn17} is t-perfect, one naturally wants to ask the same question on fork-free graphs.
The main result of the present paper is the list of forbidden t-minors (definition deferred to Section~\ref{sec:preliminaries}) of fork-free t-perfect graphs in Figure~\ref{fig:forbidden-t-minors}.
Actually, we obtain a stronger result.  A graph is \emph{strongly t-perfect} if Chv{\'a}tal's system on the graph is totally dual integral.
Note that every integral solution of Chv{\'a}tal's system corresponds to an independent set of the graph.  Thus,
it follows from the observation of Edmonds and Giles~\cite{Edmonds75} on total dual integrality that every strongly t-perfect graph is t-perfect.  The other direction remains an open problem.
In particular, we do not know whether all $P_5$-free t-perfect graphs are strongly t-perfect, though it is true for all claw-free t-perfect graphs~\cite{Bruhn10}.

\begin{figure}[h]
  \centering
  \tikzset{every path/.style={normal edge}, every node/.style={filled vertex}}
  \begin{subfigure}[b]{2.5cm}
    \centering
    \begin{tikzpicture}
      \foreach \i in {1,..., 7} {
        \draw ({(90 - 360 / 7) - \i * (360 / 7)}:1) -- ({90 - \i * (360 / 7)}:1);
      }
      \foreach \i in {1,..., 7} {
        \draw ({90 - \i * (360 / 7)}:1) -- ({(90 - 360 / 7) - (\i + 1) * (360 / 7)}:1);
      }
      \foreach \i in {1,..., 7} {
        \node[empty vertex] at ({(90 - 360 / 7) - \i * (360 / 7)}:1) {};
      }
    \end{tikzpicture}
    \caption{$C_{7}^{2}$}
  \end{subfigure}
  \begin{subfigure}[b]{2.5cm}
    \centering
    \begin{tikzpicture}
      \foreach \i in {1,..., 10} {
        \draw ({\i * (360 / 10)}:1) -- ({36 + \i * (360 / 10)}:1);
      }
      \foreach \i in {1,..., 5} {
        \draw ({\i * (360 / 5)}:1) -- ({72 + \i * (360 / 5)}:1);
      }
      \foreach \i in {1,..., 5} {
        \draw ({36 + \i * (360 / 5)}:1) -- ({-36 + \i * (360 / 5)}:1);
      }
      \foreach \i in {1,..., 5} {
        \node[empty vertex] (v\i) at ({144 + \i * (720 / 5)}:1) {};
      }
      \foreach \i in {1,..., 5} {
        \node[empty vertex] (v\i) at ({-36 + \i * (720 / 5)}:1) {};
      }
    \end{tikzpicture}
    \caption{$C_{10}^{2}$}
  \end{subfigure}
  \begin{subfigure}[b]{5cm}
    \centering
    \begin{tikzpicture}
      \foreach \i in {1,..., 3} {
        \draw ({90 - \i * (360 / 3)}:1) -- ({210 - \i * (360 / 3)}:1);
      }
      \foreach \i in {1,..., 3} {
        \draw ({90 - \i * (360 / 3)}:1) -- (0:0);
      }
      \foreach \i in {1,..., 3} {
        \node[empty vertex] at ({90 - \i * (360 / 3)}:1) {};
      }
      \node[empty vertex] at (0:0) {};
    \end{tikzpicture}
    \,
    \begin{tikzpicture}
      \foreach \i in {1,..., 5} {
        \draw ({18 - \i * (360 / 5)}:1) -- ({90 - \i * (360 / 5)}:1);
      }
      \foreach \i in {1,..., 5} {
        \draw ({18 - \i * (360 / 5)}:1) -- (0:0);
      }
      \foreach \i in {1,..., 5} {
        \node[empty vertex] at ({18 - \i * (360 / 5)}:1) {};
      }
      \node[empty vertex] at (0:0) {};
    \end{tikzpicture}
    $\cdots$
    \caption {$W_{2k+1}$ (odd wheels)}
  \end{subfigure}
  \caption{Some minimally t-imperfect graph.
  }
  \label{fig:forbidden-t-minors}
\end{figure}
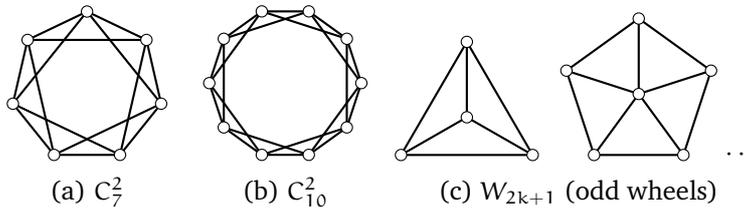

\begin{theorem}
  \label{thm:main}
  Let $G$ be a fork-free graph.  The following statements are equivalent:
  \begin{enumerate}[i)]
  \item $G$ is t-perfect.
  \item $G$ is strongly t-perfect.
  \item $G$ does not contain a $C_{7}^{2}$, a $C_{10}^{2}$, or any odd wheel as a t-minor.
  \end{enumerate}
\end{theorem}

Since strong t-perfection implies t-perfection and all the graphs in Figure~\ref{fig:forbidden-t-minors} are known to be t-imperfect~\cite{Bruhn12, schrijver-03}, we focus on showing that iii) implies ii).  To demonstrate strong t-perfection, we use the duality of linear programming: for any weighting of the vertex set, there is a certain cover of the vertex set whose cost equals the maximum weight of independent sets.  Let $G$ be a fork-free graph that does not contain a $C_{7}^{2}$, a $C_{10}^{2}$, or any odd wheel as a t-minor.
The graph must be strongly t-perfect if it is also claw-free~\cite{Bruhn12}, or if it is perfect (note that $W_{3}$ is $K_{4}$)~\cite{chvatal-75-graph-polytopes, lovasz-72-perfect-graphs}.
Moreover, if $G$ is not perfect, then the nonexistence of $K_4$'s and $C^2_7$'s forces it to contain an odd hole.
We may hence assume without loss of generality that $G$ contains a claw and an odd hole.
We show that every odd hole $H$ must be a $C_5$, and every other vertex has either exactly two consecutive neighbors or three nonconsecutive neighbors on $H$.  Based on the adjacency to vertices on $H$, we can partition $V(G)\setminus V(H)$ into a few sets.  A careful inspection of the edges among them shows that there always exists a cover.  Therefore, the graph is strongly t-perfect.  Our structural study toward Theorem~\ref{thm:main} enables us to develop polynomial-time algorithms for recognizing and coloring fork-free t-perfect graphs.

\begin{theorem}
  \label{thm:fork-free-recognize}
  Given a fork-free graph, we can decide in polynomial time whether it is t-perfect.
\end{theorem}

It is conjectured that every t-perfect graph is four-colorable~\cite{Jensen95, schrijver-03}.  We show that three colors already suffice for a fork-free t-perfect graph.

\begin{theorem}
  \label{thm:fork-free-coloring}
  Let $G$ be a fork-free graph.  If $G$ is t-perfect, then the chromatic number of $G$ is at most three, and an optimal coloring can be found in polynomial time.
\end{theorem}

We would like to ask for characterizations of $P_6$-free t-perfect graphs and $E$-free t-perfect graphs, where the $E$ graph is obtained by attaching a private neighbor to the middle vertex of a $P_5$.  The claw graph, the fork graph, and the $E$ graph are also known as the $S_{1, 1, 1}$ graph, the $S_{1, 1, 2}$ graph, and the $S_{1, 2, 2}$ graph, respectively.
Very recently, Grzesik et al.~\cite{grzesik-22-independent-p6-free} presented a polynomial-time algorithm for the independent set problem on $P_6$-free graphs.
Finding a polynomial-time algorithm for the independent set problem on $E$-free graphs is a major open problem.  Characterizing t-perfect graphs inside this class may shed some light on its solution.

\section{Graphs with no fork, $C_{7}^{2}$, $C_{10}^{2}$, odd wheel as a t-minor}
\label{sec:preliminaries}
All the graphs discussed in this paper are finite and simple. We use $V(G)$ to denote the vertex set of a graph $G$.
The set of \emph{neighbors} of $u$ is denoted as $N(u)$, whose cardinality is the \emph{degree} of $u$. 
For a proper subset $U$ of $V(G)$, we use $G-U$ to denote the graph obtained from $G$ by deleting all vertices in $U$ and all edges that have an end in $U$.  We use $G-u$ as a shorthand for $G-\{u\}$.
We say that $G$ \emph{contains} another graph $F$, or $F$ is an \emph{induced subgraph} of $G$, if $F$ can be obtained from $G$ by deleting vertices; otherwise, $G$ is \emph{$F$-free}.
For a set $\mathcal{F}$ of graphs, a graph $G$ is \emph{$\mathcal{F}$-free} if $G$ is {$F$-free} for every $F\in \mathcal{F}$.
The \emph{complement} of a graph $G$ has the same vertex set as $G$ and two vertices are adjacent if and only if they are not adjacent in $G$.  A \emph{clique} is a set of pairwise adjacent vertices, and an \emph{independent set} is a set of vertices that are pairwise nonadjacent.

For $\ell \geq 3$, we use $C_{\ell}$, $K_{\ell}$, and $P_{\ell}$ to denote the cycle graph, the complete graph, and the path graph, respectively, on $\ell$ vertices.
We use $W_{\ell}$ to denote the wheel graph, which is obtained from a $C_{\ell}$ by adding a new vertex and making it adjacent to all the vertices on the $C_{\ell}$.  Note that a $W_{3}$ is a $K_{4}$. A \emph{hole} is an induced cycle with $\ell \geq 4$. An $\ell$-cycle, $\ell$-hole, or $\ell$-wheel is odd if $\ell$ is odd; note that an odd wheel has an even number of vertices.
When we talk about the indices of vertices on an $\ell$-hole, they are always understood as modulo $\ell$.

Let $v$ be a vertex in $G$.  If $N(v)$ is an independent set, then a \emph{t-contraction} on $v$ is the operation of contracting $N(v) \cup \{v\}$ into a single vertex~\cite{gerards1998}.
See Figure~\ref{fig:t-contraction} for examples.
A graph $G'$ is a \emph{t-minor} of $G$ if $G'$ can be obtained from an induced subgraph of $G$ by a sequence of t-contractions.
Trivially, every induced subgraph of $G$ is a t-minor of $G$.
We focus on the family of fork-free graphs that does not contain $C_{7}^{2}$, $C_{10}^{2}$, or any odd wheel as a t-minor.  Since any t-minor of a fork-free graph is fork-free (the proof is omitted because it is not used in the present paper), it is precisely the family of graphs forbidding fork, $C_{7}^{2}$, $C_{10}^{2}$, and odd wheels as t-minors.

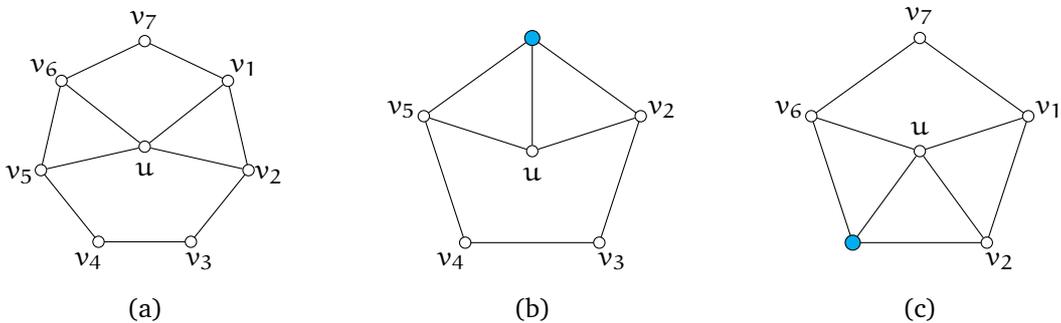
\begin{figure}[h]
  \centering
  \begin{subfigure}[b]{0.3\linewidth}
    \centering
    \begin{tikzpicture}[scale=1.4]
      \foreach \i in {0,..., 6} {
        \draw[thin] ({90 - (\i + 1) * (360 / 7)}:1) -- ({90 - \i * (360 / 7)}:1);
      }
      \foreach \i in {1,..., 7} {
        \node[empty vertex] (v\i) at ({90 - \i * (360 / 7)}:1) {};
        \node at ({90 - \i * (360 / 7)}:1.2) {$v_{\i}$};
      }
      \node[empty vertex, "$u$" below] (u) at (0, 0) {};
      \foreach \i in {1, 2, 5, 6} {
        \draw (u) -- (v\i);
      }
    \end{tikzpicture}
    \caption{}
  \end{subfigure}
  \;
  \begin{subfigure}[b]{0.3\linewidth}
    \centering
    \begin{tikzpicture}[scale=1.5]
      \foreach \i in {0,..., 4} {
        \draw[thin] ({90 - (\i + 1) * (360 / 5)}:1) -- ({90 - \i * (360 / 5)}:1);
      }
      \foreach \i in {2,..., 5} {
        \node[empty vertex] (v\i) at ({162 - \i * (360 / 5)}:1) {};
        \node at ({162 - \i * (360 / 5)}:1.2) {$v_{\i}$};
      }
      \node[filled vertex] (v1) at (90:1) {};
      \node[empty vertex, "$u$" below] (u) at (0, 0) {};
      \foreach \i in {1, 2, 5} {
        \draw (u) -- (v\i);
      }
    \end{tikzpicture}
    \caption{}
  \end{subfigure}
  \;
  \begin{subfigure}[b]{0.3\linewidth}
    \centering
    \begin{tikzpicture}[scale=1.5]
      \foreach \i in {0,..., 4} {
        \draw[thin] ({90 - (\i + 1) * (360 / 5)}:1) -- ({90 - \i * (360 / 5)}:1);
      }
      \foreach[count=\j] \i in {6, 7, 1, 2} {
        \node[empty vertex] (v\i) at ({234 - \j * (360 / 5)}:1) {};
        \node at ({234 - \j * (360 / 5)}:1.2) {$v_{\i}$};
      }
      \node[filled vertex] (v5) at (-126:1) {};
      \node[empty vertex, "$u$"] (u) at (0, 0) {};
      \foreach \i in {1, 2, 5, 6} {
        \draw (u) -- (v\i);
      }
    \end{tikzpicture}
    \caption{}
  \end{subfigure}
  \caption{A t-contraction on $v_7$ or $v_4$ in (a) leads to the graphs in (b) and (c), respectively, where the shadow vertices are the newly created.  A further t-contraction on a degree-two vertex in (b) or (c) ends with $K_4$.}
  \label{fig:t-contraction}
\end{figure}

As mentioned in Section~\ref{sec:intro}, we assume that the graph contains an odd hole and a claw.  This section is devoted to a structural study of such graphs.
Note that a fork-free graph that contains an odd hole but does not contain $C_{7}^{2}$, $C_{10}^{2}$, or any odd wheel as a t-minor satisfies conditions of all the statements.  We use weaker conditions when the statement may be of independent interest.
The first observation is about the neighborhood of an outside vertex on an odd hole.

\begin{proposition}
  \label{prop:neighbors-on-hole}
  Let $G$ be a graph containing an odd hole $H$ and $u$ a vertex in $V(G) \setminus V(H)$.
  \begin{enumerate}[i)]
  \item If $u$ has exactly one neighbor on $H$, then $G$ contains a fork.
  \item If $u$ has exactly two neighbors on $H$ and they are not consecutive on $H$, then $G$ contains a fork.
  \item If $u$ has exactly three neighbors on $H$ and they are consecutive on $H$, then $K_4$ is a t-minor of $G$.
  \item If $u$ has exactly four neighbors on $H$ and they form one or two paths on $H$, then $K_4$ is a t-minor of $G$.
  \end{enumerate}
\end{proposition}
\begin{proof}
  For assertions i) and ii), we number the vertices on $H$ as $v_1, v_2, \ldots$.  Suppose without loss of generality that $u$ is adjacent to $v_{3}$.  Then $u$ is adjacent to neither $v_{2}$ nor $v_{4}$.
  There is no other neighbor of $u$ on $H$ in i).
  In ii), $u$ cannot be adjacent to both $v_1$ and $v_5$; we may assume that $u$ is not adjacent to $v_1$.
  Then $\{v_{3}, v_{4},u , v_{2}, v_{1}\}$ forms a fork.\footnote{When we list the vertices of a (potential) fork, we always put the degree-three vertex first, followed by its three neighbors, the last of which has degree two.}

  For assertions iii) and iv), we focus on the subgraph $G'$ induced by $V(H) \cup \{u\}$; see Figure~\ref{fig:t-contraction}.  Note that any vertex in $V(H)\setminus N(u)$ has only two neighbors in $G'$, and they are not adjacent.  We do induction on the length of $H$.  In the base case, $|H| = 5$.  (Note that in this case, if $u$ has four neighbors on $H$, then they must be consecutive.)
  A t-contraction on a vertex in $V(H)\setminus N(u)$ leads to a $K_{4}$.  We now consider that $|H| > 5$.
  We apply a t-contraction on a vertex $v$ in $V(H)\setminus N(u)$, which shortens $H$ into a shorter odd hole, denoted by $H'$.
  The length of $H'$ is two shorter than $H$.
  If the neighbors of $u$ on $H$ are consecutive, then the two neighbors of $v$ cannot be both adjacent to $u$ (note that $|H| \ge 7$).  Thus, $u$ has the same number of neighbors on $H'$ as on $H$, and they remain consecutive.
  In the rest, $u$ must have four neighbors on $H$, and they form two paths.  If the two neighbors of $v$ are both adjacent to $u$, then $u$ has three consecutive neighbors on $H'$.  Otherwise, $u$ still has exactly four neighbors on $H'$ and they form one or two paths on $H'$.  Thus, the statements hold by induction.
\end{proof}

The following statement further extends Proposition~\ref{prop:neighbors-on-hole}(ii).  Each edge of an odd hole can only have one private neighbor, which is not adjacent to any other vertex on the hole.

\begin{proposition}
  \label{prop:size-constraint}
  Let $G$ be a $\{K_{4}, \mathrm{fork}\}$-free graph containing an odd hole $H$.
  For any two vertices on $H$, at most one of their common neighbors is adjacent to only two vertices on $H$.
\end{proposition}
\begin{proof}
  Suppose for contradiction that there are two distinct vertices $x$ and $y$ such that they have the same pair of neighbors on $H$.  By Proposition~\ref{prop:neighbors-on-hole}(ii), the neighbors of $x$ on $H$ have to be consecutive.  We number the vertices on $H$ as $v_1, v_2, \ldots$ such that $x$ is adjacent to $v_1$ and $v_2$.  Since $G$ is $K_4$-free, $x$ and $y$ are not adjacent.  Then $\{v_{2}, x, y, v_{3}, v_{4}\}$ forms a fork, a contradiction.
\end{proof}

The existence of claws forbids the possibility that each edge of an odd hold $H$ has precisely two neighbors on it.

\begin{proposition}
  \label{lem:reduce-to-claw-free}
  Let $G$ be a fork-free graph containing an odd hole $H$.  The graph $G$ is claw-free if
  \begin{enumerate} [i)]
  \item $G$ contains neither $K_{4}$ nor $W_{5}$; and
  \item every vertex $v\in V(G)\setminus V(H)$ has either zero or two neighbors on $H$.  %
  \end{enumerate}
\end{proposition}
\begin{proof}
  Suppose that $G$ satisfies both conditions i) and ii).  We number the vertices on $H$ as $v_1, v_2, \ldots$.
  Since $G$ is a fork-free graph, if a vertex $v\in V(G)\setminus V(H)$ has two neighbors on $H$, then they are consecutive by Proposition~\ref{prop:neighbors-on-hole}(ii).
  Thus, for each $i$, a neighbor of $v_i$ not on $H$ is adjacent to either $v_{i-1}$ or $v_{i+1}$.
  By Proposition~\ref{prop:size-constraint}, the degree of $v_i$ is at most four, and it cannot be the center of a claw.
  Suppose for contradiction that $G$ contains a claw.  We take a claw $T$ of $G$ whose center has the shortest distance to $H$, denoted as $d$, among all claws of $G$.  Note that $d\ge 1$.  Let the vertex set of $T$ be $\{c, x, y, z\}$, where $c$ is the center of $T$.

  Case 1, $d = 1$. The vertex $c$ is adjacent to $H$, and by assumption, it has exactly two neighbors on $H$.  We first note that we can choose $T$ to intersect $H$.  Suppose that none of $x$, $y$, and $z$ is on $H$, and let $v$ be a neighbor of $c$ on $H$. Since the degree of $v$ is at most four, $v$ has at most one neighbor, say $z$, in $\{x, y, z\}$.  Then $\{c, x, y, v\}$ is another claw. In the rest, without loss of generality, let the two neighbors of $c$ on $H$ be $v_{1}$ and $v_{2}$, where $z = v_{2}$. Since $\{c, x, y, v_{2},v_{3}\}$ cannot induce a fork, at least one of $x$ and $y$ is adjacent to $v_3$.  Since neither $x$ nor $y$ is adjacent to $v_2$, they cannot be both adjacent to $v_3$. We may assume that $y$ is adjacent to $v_{3}$, hence to $v_{4}$ as well but no other vertex on $H$.  Since $\{c, x, v_{2}, y, v_{4}\}$ does not induce a fork, $x$ has to be adjacent to $v_4$ as well, and its other neighbor on $H$ is $v_5$.  But then $\{c, y, z, x, v_{5}\}$ induces a fork, a contradiction.

  Case 2, $d = 2$. We may assume that there is a common neighbor $p$ of $c$ and $v_{1}$, and the other neighbor of $p$ on $H$ is $v_{2}$. (Note that $c$ is not adjacent to $v_{1}$.)
  \begin{itemize}
  \item Subcase 2.1, $p$ has two or more neighbors in $\{x, y, z\}$, say $x$ and $y$. By the selection of $T$, there cannot be any claw in $G$ that has $p$ as the center.  Thus, $x$ is adjacent to either $v_{1}$ or $v_{2}$, and so is $y$.  Either $c x v_{1} v_{2} y$ or $c x v_{2} v_{1} y$ is a five-hole, and $p$ is adjacent to all vertices on it.  Therefore, $G$ contains a $W_{5}$, a contradiction.
  \item Subcase 2.2, $p$ has at most one neighbor in $\{x, y, z\}$.   (We are in this sub-case when $p$ is one of $\{x, y, z\}$.)
    Assume without loss of generality that $p$ is adjacent to neither $x$ nor $y$.
    Since $\{c, x, y, p, v_{1}\}$ does not induce a fork, $v_{1}$ is adjacent to at least one of $x$ and $y$.  On the other hand, $v_{2}$ cannot be adjacent to both $x$ and $y$. We may assume that $y$ is adjacent to $v_{2}$; note that the other neighbor of $y$ on $H$ has to be $v_3$.  Since $\{c, p, x, y, v_{3}\}$ does not induce a fork, $x$ is adjacent to $v_3$ as well; note that the other neighbor of $x$ on $H$ has to be $v_4$. But then $\{c, p, y, x, v_{4}\}$ induces a fork, a contradiction.
  \end{itemize}

  Case 3, $d \geq 3$. Let $c u_1 u_2 \cdots$ be a shortest path from $c$ to $H$.  Note that no vertex in $\{x, y, z\}$ is adjacent to $u_i$ with $i > 2$.
  \begin{itemize}
  \item Subcase 3.1,
    $u_1$ has at most one neighbor in $\{x, y, z\}$.   (We are in this sub-case when $u_1$ is one of $\{x, y, z\}$.)
    Assume without loss of generality that $u_1$ is adjacent to neither $x$ nor $y$.
    Since $\{c, x, y, u_{1}, u_{2}\}$ does not induce a fork, $u_{2}$ is adjacent to at least one of $x$ and $y$, say $y$.  But then $\{u_{2}, u_{1}, u_{3}, y\}$ induces a claw, and its center $u_2$ has a shorter distance to $H$ than $c$, a contradiction to the selection of $T$.
  \item Subcase 3.2, $u_1$ has two or more neighbors in $\{x, y, z\}$, say $x$ and $y$.  Note that $u_1$ cannot be adjacent to $z$; otherwise $\{u_1, x, y, z\}$ induces a claw, which contradicts the selection of $T$. If $u_2$ is adjacent to only $x$ in $\{x, y, z\}$, then $\{c, y, z, x, u_2\}$ induces a fork. If $u_2$ is adjacent to two vertices in $\{x, y, z\}$, then these two vertices, together with $u_2$ and $u_3$, form a claw that is closer to $H$ than $T$.  If $u_2$ is adjacent to neither $x$ nor $y$, then $\{u_1, u_2, x, y\}$ induces a claw that is closer to $H$ than $T$.
  \end{itemize}
  There, there cannot be a claw in $G$.
\end{proof}

The somewhat conflicting requirements in Propositions~\ref{prop:neighbors-on-hole} and \ref{lem:reduce-to-claw-free} exclude odd holes longer than five, and force every five-hole to be dominating (i.e., every vertex in the same component has neighbors on this hole).

\begin{proposition}
  \label{lem:fix-longer-hole}
  Let $G$ be a fork-free graph containing an odd hole $H$.  If $G$ contains a claw and does not contain any odd wheel as a t-minor, then $|H| = 5$, and every vertex in the same component as $H$ is adjacent to $H$.
\end{proposition}
\begin{proof}
  We number the vertices of $H$ as $v_{1}, \dots, v_{\ell}$, where $\ell = 2k+1$.
  Since $G$ contains a claw, by Proposition~\ref{prop:neighbors-on-hole} and Proposition~\ref{lem:reduce-to-claw-free}, we can find a vertex $u$ that has three or more neighbors on $H$.

  We first show $|H| = 5$ by contradiction.
  On the other hand, $u$ has a non-neighbor on $H$ because $G$ is free of odd wheels.
  We may assume without loss of generality that $u$ is adjacent to $v_4$ but not $v_5$.
  We argue that $u$ is adjacent to $v_3$ by contradiction.
  If $u$ is not adjacent to $v_{3}$, then $u$ must be adjacent to $v_{2}$ because $\{v_4,v_5, u,v_3,v_2\}$ does not induce a fork.
  By symmetry, $u$ is adjacent to $v_{6}$. But then dependent on the adjacency between $u$ and $v_{1}$, either $\{v_4,v_3, v_5,u,v_1\}$ or $\{u,v_4,v_6,v_2,v_1\}$ induces a fork.   Thus, $u$ is adjacent to $v_{3}$.

  If $u$ has precisely three neighbors and the only other neighbor of $u$ on $H$ is $v_2$, then $G$ contains $K_4$ as a t-minor by Proposition~\ref{prop:neighbors-on-hole}(iii).  Therefore, there is a neighbor $v_i$ of $u$ with $i\not\in \{2, 3, 4\}$.  We traverse $H$ from $v_4$, $v_5$ till we meet the next neighbor of $u$; let it be $v_{i}$. Since $\ell \ge 7$ and $i \ne 2$, one of $v_3$ and $v_4$ is nonadjacent to all the vertices in $\{v_{i-1},v_{i},v_{i+1}\}$. Since neither $\{v_i, v_{i-1}, v_{i+1},u,v_4\}$ nor $\{v_i, v_{i-1}, v_{i+1},u,v_3\}$ induces a fork, $u$ has to be adjacent to $v_{i+1}$ as well. If $u$ has precisely four neighbors on $H$, namely, $v_3$, $v_4$, $v_i$, and $v_{i+1}$, then $G$ contains $K_4$ as a t-minor by Proposition~\ref{prop:neighbors-on-hole}(iv).  Therefore, $u$ has at least five neighbors on $H$.  As a result, $6 \le i \le \ell$. If $i > 6$, then $\{u, v_{4}, v_{j}, v_{i}, v_{i-1}\}$, where $v_j$ is another neighbor of $u$ on $H$, induces a fork (note that $v_j$ cannot be adjacent to $v_{4}$, $v_{i-1}$, or $v_{i}$). In the rest, $i=6$.  If $v_5$ is the only non-neighbor of $u$ on $H$, then $G$ contains an odd wheel as a t-minor. Hence, $u$ has at least one neighbor and one non-neighbor in $\{v_8, v_9, \ldots, v_\ell, v_1, v_2\}$.  We can find a $j$ such that $u$ is adjacent to precisely one of $\{v_j, v_{j+1}\}$. But then $\{u,v_{4},v_{6}, v_j, v_{j+1}\}$ induces a fork (note that there cannot be any edge between $v_{4},v_{6}$ and $v_j, v_{j+1}$).
  Therefore, the length of $H$ has to be five.

  By Proposition~\ref{prop:neighbors-on-hole}(iii, iv), the vertex $u$ has three nonconsecutive neighbors on $H$.  Assume without loss of generality that they are $v_1$, $v_3$, and $v_4$.  Let $X = V(H) \cup \{u\}$.
  Suppose that there exists a vertex $x$ that is in the same component as $H$ but has no neighbor on $H$.  We can find a shortest path $x_{1} x_{2} \dots x_{p}$ between $x_{p} = x$ and $H$; hence, $x_{1}$ is the only common vertex of this path and $H$.  Note that $p\ge 3$, and the vertex $x_{3}$ has no neighbor on $H$.  If $x_3$ is adjacent to $u$, then $\{v_{1}, v_{2}, v_{5}, u, x_{3}\}$ forms a fork with $v_1$ as the degree-three vertex.  Therefore, $N(x_{3})\cap X = \emptyset$.  Since $x_{2}$ is adjacent to both $x_{3}$ and $H$, it is adjacent to all the vertices on $H$ according to Lozin and Milani\v{c}~\cite[Lemma 1]{Lozin08}.  But then we have an odd wheel, which is impossible.
\end{proof}

We fix a five-hole $H$ and number its vertices as $v_{1}, \dots, v_{5}$.
For $i = 1, \dots, 5$, let $U_{i}$ be the set of common neighbors of $v_{i+2}$ and $v_{i+3}$.  They are independent sets because $G$ is $K_4$-free.
By Propositions~\ref{prop:neighbors-on-hole}--\ref{lem:fix-longer-hole}, together with $V(H)$, the five sets partition $V(G)$.  An independent set is \emph{maximal} if it is not a subset of any other independent set.

\begin{proposition}\label{prop:independent-set-2}
  Let $G$ be a $K_{4}$-free graph containing a five-hole $H$. 
  For $i = 1, \ldots, 5$, the set $\{v_{i-1},v_{i+1}\} \cup U_{i}$ is a maximal independent set of $G$.
  Moreover, they are all the maximal independent sets of $G$ that contains two vertices from $H$.
\end{proposition}
\begin{proof}
  Let $S$ be a maximal independent set of $G$.
  Since $H$ is a $C_5$, it follows 
  $|S\cap V(H)| \le 2$.
  If $S$ contains precisely two vertices from $H$, they have to be $v_{i-1}$ and $v_{i+1}$ for some $i$.  By the definitions of $U_i$, we have $S\setminus V(H)\subseteq U_i$.  Thus, $S \subseteq \{v_{i-1},v_{i+1}\} \cup U_{i}$.
  Since there is no edge among vertices in $U_i$, it must hold by equality by the maximality of $S$.
\end{proof}

If a vertex in $U_{i}$ has another neighbor on $H$, then it has to be $v_i$ by Propositions~\ref{prop:neighbors-on-hole}(iii).
We can thus partition $U_{i}$ into $U_{i}^{+} = U_{i}\cap N(v_i)$ and $U_{i}^{-} = U_{i}\setminus N(v_i)$.
For any vertex $x\in U_{i}^+$, the set $\{v_{i}, v_{i-1}, v_{i+1}, x\}$ induces a claw.
According to Proposition~\ref{prop:size-constraint}, $|U_{i}^{-}| \leq 1$ for $i = 1, \dots, 5$.
By Proposition~\ref{lem:reduce-to-claw-free}, $\bigcup_{i=1}^5 U_{i}^{+}$ is not empty.
We summarize the adjacency relations among the parts in the following proposition when $U_{i}^+ \neq \emptyset$.
Two disjoint vertex sets are \emph{complete} if all the edges between them are present.

\begin{proposition}
  \label{prop:fork-free-results}
  Let $G$ be a $\{K_{4}, W_5, \mathrm{fork}\}$-free graph containing a five-hole $H$.  If any vertex in $V(G)\setminus V(H)$ has either two consecutive neighbors or three nonconsecutive neighbors on $H$, and $U_{i}^{+}$ is nonempty for some $i = 1, \ldots, 5$, then
  \begin{enumerate} [i)]
    \item \label{lb:fork-free-i} $U_{i}$ is complete to  $U_{i-2} \cup U_{i+2}$;
    \item \label{lb:fork-free-ii} $U_{i}$ is complete to $U_{i-1}^{-} \cup U_{i+1}^{-}$;
  \item \label{lb:fork-free-iii} $U_{i+1}^{-}$ is complete to $U_{i+2}^{-}$;  
     \item \label{lb:fork-free-v} at least one of $U_{i+2}$ and $U_{i-2}$ is empty; and
     \item \label{lb:fork-free-iv} a vertex in $U_{i}^{+}$ has at most one non-neighbor in $U_{i-1}^{+}$ and at most one non-neighbor in $U_{i+1}^{+}$. 
  \end{enumerate} 
\end{proposition}
\begin{proof}
  We show the statements for $i = 3$; they hold for other indices by symmetry.

  (i) Let $x$ be an arbitrary vertex in $U_{3}$ and $y$ an arbitrary vertex in $U_{5}$.
  Suppose first that $x\in U_{3}^{+}$.
  By definition, $x$ is adjacent to both $v_{1}$ and $v_{3}$ but not $v_{4}$,
  and $y$ is adjacent to $v_{3}$ but neither $v_{1}$ nor $v_{4}$.
  They have to be adjacent as otherwise $\{v_{3}, y, v_{4}, x, v_{1}\}$ forms a fork.
  Thus, $U_{3}^{+}$ is complete to $U_5$, and a similar argument implies that $U_{3}$ is complete to $U_5^{+}$.
  The only remaining case is when $x\in U_{3}^{-}$ and $y\in U_{5}^{-}$ (we have nothing to show if one or both of them are empty).
  We take an arbitrary vertex $x'\in U_{3}^{+}$, which is nonempty by assumption.  Note that $x$ is not adjacent to $x'$, and we have seen above that $x'$ is adjacent to $y$.
  By definition, both $x$ and $x'$ are adjacent to $v_{5}$ and neither is adjacent to $v_{4}$.
Thus, $x$ and $y$ must be adjacent as otherwise $\{v_{5}, v_{4}, x, x', y\}$ forms a fork.  
    A symmetric argument applies to $U_{3}$ and $U_{1}$.
    
    (ii) Let $x$ be an arbitrary vertex in $U_{3}$.
    The statement holds vacuously for $U_{4}^{-}$ when it is empty.  Assume that $U_{4}^{-}\ne\emptyset$ and $y$ be the vertex in $U_{4}^{-}$.
    By definition, $y$ it is adjacent to $v_{2}$ but none of $v_{3}$, $v_{4}$, and $v_{5}$.    
    If $x$ is in $U_{3}^{+}$, then $x$ must be adjacent to $y$ as otherwise $\{v_{3}, v_{4}, x, v_{2}, y\}$ forms a fork.
    In the remaining case, $x \in U_{3}^{-}$.
    We take an arbitrary vertex $x'\in U_{3}^{+}$, which is nonempty by assumption.
    By the argument above, $x'$ is adjacent to $y$.
    The vertices $x$ and $y$ must be adjacent as otherwise $\{v_{5}, v_{4}, x, x', y\}$ forms a fork.
    A symmetric argument implies that $U_3$is complete to $U_{2} \setminus U_{2}^{+}$.     

    iii)
    This assertion holds vacuously for when one or both
    of $U_{4}^{-}$ and $U_{5}^{-}$ are empty.  Hence we may assume otherwise.  For $j = 4, 5$, let $u_{j}^{-}$ be the vertex in $U_{j}^{-}$.
    We take an arbitrary vertex $x\in U_{3}^{+}$, which is nonempty by assumption.
    By definition, the vertex $x$ is adjacent to $v_{5}$ but not $v_{4}$, the vertex $u_{4}^{-}$ is adjacent to neither $v_{5}$ nor $v_{4}$, and the vertex $u_{5}^{-}$ is adjacent to neither $v_{4}$ nor $v_{5}$.
    By assertions (i, ii), $x$ is adjacent to both $u_{4}^{-}$ and $u_{5}^{-}$.  Therefore, $z$ must be adjacent to $y$ as otherwise $\{x,y,z,v_{5},v_{4}\}$ forms a fork.

    (iv) Suppose for contradiction that neither $U_{1}$ nor $U_{5}$ is empty.  We pick three arbitrary vertices
    $u_{1}$, $u_{3}^{+}$, and $u_{5}$ from $U_{1}$, $U_{3}^{+}$, and $U_{5}$, respectively.  By assertion~(\ref{lb:fork-free-i}), $u_{3}^{+}$ is adjacent to both $u_{1}$ and $u_{5}$.
    If $u_{1}$ and $u_{5}$ are adjacent, then $\{u_{3}^{+},u_{5},u_{1},v_{3}\}$ is a clique, a contradiction to that $G$ is $K_4$-free.  In the rest, $u_{1}$ is not adjacent to $u_{5}$.  The vertex $v_{1}$ must be adjacent to $u_{1}$ as otherwise $\{u_{3}^{+}, u_{5},v_{1},u_{1},v_{4}\}$ forms a fork.  By symmetry, $v_{5}$ is adjacent to $u_{5}$.   But then $u_{3}^{+}$ has five neighbors on the hole $u_{5} v_{5} v_{2} u_{1} v_{3}$, contradicting that $G$ is $W_5$-free.  

    (v) Let $x$ be an arbitrary vertex in $U_{3}^{+}$.
    If there are two distinct vertices $y$ and $y'$ in $U_{4}^{+}\setminus N(x)$, then $\{v_{2}, y, y', v_{3}, x\}$ forms a fork.
    A symmetric argument applies to $U_{2}^{+}$.
\end{proof}

The rest of this section is about coloring.
For a positive integer $k$,  we call $G$ \emph{$k$-colorable} if we can partition $V(G)$ into $k$ independent sets.  The smallest $k$ such that $G$ is $k$-colorable is the \emph{chromatic number} of $G$.

\begin{lemma}
  \label{lem:optimal-coloring}
  Let $G$ be a fork-free graph.  If $G$ does not contain a $C_{7}^{2}$, a $C_{10}^{2}$, or any odd wheel as a t-minor, then the chromatic number of $G$ is at most three, and an optimal coloring of $G$ can be found in polynomial time.
\end{lemma}
\begin{proof}
    Since $G$ does not contain $K_4$, it cannot contain the complement of any odd hole longer than seven.  It does not contain $C^2_7$, which is the complement of $C_7$.  Therefore, if $G$ does not contain any odd hole, then $G$ is perfect, and
    we can use the algorithm of Chudnovsky et al.~\cite{Chudnovsky17} to find an optimal coloring.
  The chromatic number of $G$ is at most three because it is equal to the order of the maximum cliques \cite{chudnovsky-06-strong-perfect-graphs-theorem}, which is at most three because $G$ is $K_4$-free.
  Otherwise, $G$ contains an odd hole, and thus its chromatic number is at least three.  Thus, it suffices to find a three coloring, i.e., a partition of $V(G)$ into three (not necessarily maximal) independent sets.
  If $G$ is claw-free, then we can use the algorithm of Bruhn and Stein~\cite{Bruhn12} to find an optimal coloring.  In the rest, $G$ contains a claw.
  The algorithm now finds a five-hole $H$, and partition the vertex set $V(G)\setminus V(H)$ according to their adjacencies with $H$.  We may number the vertices on $H$ such that $U_1^+$ is nonempty and $U_{4}$ is empty.  This is possible because of Proposition~\ref{prop:fork-free-results}(\ref{lb:fork-free-v}) with $i=1$.

  If $U_{3}$ is empty, then we partition $V(G)$ into three independent sets $U_{5} \cup \{v_{4}\}$, $U_{1} \cup \{v_{2}, v_{5}\}$, and $U_{2} \cup \{v_{1}, v_{3}\}$.
  If $U_5$ is empty, then we partition $V(G)$ into three independent sets that are $U_{1} \cup \{v_{5}\}$, $U_{2} \cup \{v_{1}, v_{3}\}$, and $U_{3} \cup \{v_{2}, v_{4}\}$.
  In the rest, neither $U_{3}$ nor $U_5$ is empty. 
  If $U_{5}^{+}\ne\emptyset$, then $U_{2}$ is empty because of Proposition~\ref{prop:fork-free-results}(\ref{lb:fork-free-v}) with $i=5$.   We can partition $V(G)$ into three independent sets $U_{1} \cup \{v_{2}, v_{5}\}$, $U_{3} \cup \{v_{3}\}$, and $U_{5} \cup \{v_{1}, v_{4}\}$.
  The remaining case is when $U_{5} = U_{5}^{-}$, and we show that this cannot happen.
  Since neither $U_{1}$ nor $U_{5}$ is empty, $U_{3}^{+}$ is empty because of Proposition~\ref{prop:fork-free-results}(\ref{lb:fork-free-v}) with $i=3$.
  For $j=3, 5$, let $u_{j}^{-}$ be the only vertex in $U_{j}^{-}$.
  Let $u_{1}^{+}$ be an arbitrary vertex in $U_{1}^{+}$; it is adjacent to both $u_{3}^{-}$, by Proposition~\ref{prop:fork-free-results}(\ref{lb:fork-free-i}), and $u_{5}^{-}$, by Proposition~\ref{prop:fork-free-results}(\ref{lb:fork-free-ii}), both with $i=1$.
  If $u_{3}^{-}$ is not adjacent to $u_{5}^{-}$, then $\{u_{1}^{+}, u_{3}^{-}, v_{4}, u_{5}^{-}, v_{2}\}$ forms a fork; otherwise, $u_{1}^{+}$ has three consecutive neighbors on the hole $u_{3}^{-} v_{5} v_{4} v_{3} u_{5}^{-}$, contradicting Propositions~\ref{prop:neighbors-on-hole}(iii).
  The algorithm is thus complete

  All the induced subgraphs we need to check have a constant number of vertices.  Both algorithms we call~\cite{Chudnovsky17, Bruhn12} take polynomial time.  The rest is clearly doable in polynomial time.
  Thus, the whole algorithm runs in polynomial time.
\end{proof}

Theorem~\ref{thm:fork-free-coloring} directly follows from Lemma~\ref{lem:optimal-coloring} and Theorem~\ref{thm:main}.

\section{Strong t-perfection}
\label{sec:t-perfect}
The \emph{independent set polytope} of a graph is defined as the convex hull of the characteristic vectors of all independent sets in it.  For a graph $G$, we define another polytope $P(G)$
as the set of vectors $x\in \mathbb{R}^{V(G)}$ satisfying

\begin{equation}
  \label{eq:t-perfect}
  \begin{aligned}
    0 \leq x_v & \leq 1                  &  & \text{for every vertex } v,
    \\
    x_u + x_v  & \leq 1                  &  & \text{for every edge } u v,
    \\
    x(V(C))    & \leq \frac{|V(C)|-1}{2} &  & \text{for every odd cycle $C$ in } G,
  \end{aligned}
\end{equation}
where $x(S) = \sum_{v\in S} x_v$.
Clearly, the characteristic vector of every independent set of $G$ satisfies all the constraints in \eqref{eq:t-perfect}.  Thus, the independent set polytope is contained in $P(G)$.  While the other direction is not true in general:
the vector $x = (\frac{1}{3},\frac{1}{3},\frac{1}{3},\frac{1}{3})^{\mathsf{T}}$ is in $P(K_{4})$ but is not in the independent set polytope of $K_{4}$ because it does not satisfy the clique constraint.
A graph $G$ is \emph{t-perfect} if $P(G)$ is precisely the independent set polytope of $G$, and \emph{strongly t-perfect} if the system~(\ref{eq:t-perfect}) is totally dual integral.
It is well known that every strongly t-perfect graph is t-perfect~\cite{Edmonds75, schrijver-03}, while the other direction remains an open problem.
The $K_{4}$ is the smallest graph that is not t-perfect, hence not strongly t-perfect.
For a vertex $v \in V(G)$, the polytope $P(G-v)$ is the projection of the intersection of $P(G)$ and the hyperplane $x_{v} = 0$ on $\mathbb{R}^{V(G-v)}$. Therefore, t-perfection is preserved under vertex deletions.
Both t-perfection and strong t-perfection are also invariants with respect to t-contraction~\cite{gerards1998, Bruhn10}.

A graph is \emph{perfect} if it does not contain an odd hole or the complement of an odd hole.  A graph is perfect if and only if its independent set polytope can be determined by the following linear system \cite{padberg-74-perfect-matrices, chvatal-75-graph-polytopes}:
\begin{equation}
  \label{eq:perfect}
  \begin{aligned}
    x_v  &               \geq 0    &  & \text{for every vertex } v,
    \\
    x(K)   &      \leq 1             &  & \text{for every clique } K \text{ in } G.
  \end{aligned}
\end{equation}
Note that a triangle constraint is both a clique constraint and an odd-cycle constraint.  Moreover, the odd-cycle constraints in~\eqref{eq:t-perfect} can be restricted to induced odd cycles: those on non-induced ones are redundant.
In a $K_{4}$-free graph, the clique constraint of \eqref{eq:perfect} degenerates to vertex, edge, and triangle constraints of \eqref{eq:t-perfect}.
On the other hand, if $G$ is perfect, then it does not contain odd holes, and odd-cycle constraints of \eqref{eq:t-perfect} degenerates to triangle constraints.
Thus, \eqref{eq:t-perfect} and \eqref{eq:perfect} coincide for a $K_{4}$-free perfect graph.
If the system \eqref{eq:perfect} is totally dual integral, then $G$ is perfect~\cite{Edmonds75, chvatal-75-graph-polytopes}.
Lov\'{a}sz \cite{lovasz-72-perfect-graphs} showed that the converse is also true.
Therefore, the system~\eqref{eq:perfect} is totally dual integral if and only if $G$ is perfect; see also~\cite{schrijver-03} for more details.

\begin{proposition}[Folklore]
  \label{prop:perfect and t-perfect}
  Every $K_{4}$-free perfect graph is strongly t-perfect.
\end{proposition}

Propositions~\ref{prop:neighbors-on-hole}--\ref{lem:fix-longer-hole} can be summarized as follows.
If a fork-free graph $G$ contains a claw and an odd hole and does not contain $C_{7}^{2}$, $C_{10}^{2}$, or any odd wheel as a t-minor, then every odd hole $H$ in $G$ has length five, and satisfies the following property.

 \begin{enumerate}[($\star$)]
 \item \label{p1} A vertex in $V(G) \setminus V(H)$ has either exactly two consecutive neighbors on $H$, or three nonconsecutive neighbors on $H$.  
\end{enumerate}
Interestingly, the other direction also holds true.  
The main work of this section is to establish the following lemma.  
\begin{lemma}\label{lem:main}
  Let $G$ be a fork-free graph that contains a claw and an odd hole.  The following statements are equivalent:
  \begin{enumerate}[i)]
  \item $G$ does not contain the $C_{7}^{2}$, the $C_{10}^{2}$, or any odd wheel as a t-minor.
  \item $G$ is $\{K_4, W_{5}, C^2_7, C_{10}^{2}\}$-free, and every odd hole in $G$ has length five and satisfies \myref{}.
  \item $G$ is strongly t-perfect.
  \end{enumerate}
\end{lemma}

Before presenting the proof of Lemma~\ref{lem:main}, we use it to prove Theorem~\ref{thm:main}.

\begin{proof} [Proof of Theorem~\ref{thm:main}]
  Since strong t-perfection implies t-perfection and $C_{7}^{2}$, $C_{10}^{2}$, and all odd wheels are t-imperfect~\cite{Bruhn12, schrijver-03}, it suffices to show that if a fork-free graph does not contain $C_{7}^{2}$, $C_{10}^{2}$, or any odd wheel as a t-minor, then it is strongly t-perfect.
  Suppose that $G$ is such a graph.
  If $G$ is claw-free, then $G$ is strongly t-perfect according to Bruhn and Stein~\cite[Theorem 2]{Bruhn10} and \cite[Theorem 3]{Bruhn12}.
  Note that the complement of $C_{7}$ is $C_{7}^{2}$, and the complement of an odd hole longer than seven contains a $K_{4}$.  If $G$ does not contain an odd hole, then $G$ is perfect, and hence strongly t-perfect by Proposition~\ref{prop:perfect and t-perfect}.
  Now that $G$ contains a claw and an odd hole, it is strongly t-perfect by Lemma~\ref{lem:main}.
\end{proof}

The rest of the section is devoted to proving Lemma~\ref{lem:main}.
By \emph{duplicating} a vertex $v$ of $G$ we introduce copies of $v$ and make them adjacent to every neighbor of $v$ in $G$.
Benchetrit~\cite{Benchetrit15} proved that the class of strongly t-perfect graphs is closed under vertex duplication.
\begin{lemma}[{\cite{Benchetrit15}}]
  \label{lem:substitution}
  The graph obtained by duplicating any vertex of a strongly t-perfect graph is strongly t-perfect.
\end{lemma}

The graph $G$ is $\{$fork, $K_4, W_{5}, C^2_7, C_{10}^{2}\}$-free, contains a claw and an odd hole, and every odd hole has length five and satisfies \myref{}.
We fix a five-hole $H$, and partition the vertices $V(G)\setminus V(H)$ into $U_1, \ldots, U_5$.  Each $U_{i}$ is further partitioned into $U_{i}^{+}$ and $U_{i}^{-}$.
Recall that $|U_{i}^{-}|\le 1$ by Proposition~\ref{prop:size-constraint}.
By Proposition~\ref{prop:fork-free-results}, the main uncertain adjacencies are between $U_{i}^+$ and $U_{i+1}^+$.
Thus, if only one of $U_{i}^+$'s or two nonconsecutive of them are nonempty, then the graph has a very simple structure.
Indeed, it can be obtained from one of the small graphs (of order at most ten) in Figure~\ref{fig:untangled} by vertex duplications.

\begin{lemma}\label{lem:untangled}
  If for any $i = 1, \dots, 5$, one of $U_{i}^+$ and $U_{i+1}^+$ is empty, then $G$ is strongly t-perfect.
\end{lemma}
\begin{proof}
  We may assume without loss of generality that $U_{2}^+$ is nonempty, while all of $U_{1}^+$, $U_{3}^+$, and $U_{5}^+$ are empty.
  Every vertex in $U_{2}^+$ is adjacent to $v_2, v_4$, and $v_5$ but not $v_1$ or $v_3$ by definition.

    Suppose first that $U_{4}^+$ is also nonempty.  Then both $U_{1}$ and $U_{5}$ are empty by Proposition~\ref{prop:fork-free-results}(\ref{lb:fork-free-v}).  Thus,
    \[
      V(G) \setminus  (V(H) \cup U_{2}^+ \cup U_{4}^+) = U_{2}^-\cup U_{3}^-\cup U_{4}^-.
    \]
    By Proposition~\ref{prop:fork-free-results}(i, ii), all the edges between $U_{2}^+$ and $U_{3}^-\cup U_{4}$ are present.
    Thus, all vertices in $U_{2}^+$ have the same neighborhood in $V(G)\setminus U_{2}^+$.
    A symmetric argument applies to $U_{4}^+$.
    Let $G_1$ be the graph in Figure~\ref{fig:untangled}(a), where for $i = 2, 3, 4$, the optional vertices $u_{i}^-$ exists if and only if $U_{i}^- \ne \emptyset$.
    It is easy to verify that $G_1$ is strongly t-perfect if it satisfies the condition of Lemma~\ref{lem:main}(ii).
    We duplicate $u_{2}^+$ of $G_1$ with $|U_{2}^{+}|$ vertices, and then duplicate $u_{3}^{+}$ in the resulted graph with $|U_{3}^{+}|$ vertices.  The final result is $G$.  Therefore, $G$ is strongly t-perfect by Lemma~\ref{lem:substitution}.
    
    In the rest, $U_{4}^+$ is empty.  We may assume without loss of generality that $U_{5} = \emptyset$.
    Then
    \[
      V(G) \setminus (V(H) \cup U_{2}^+) = U_{1}^-\cup U_{2}^-\cup U_{3}^-\cup U_{4}^-.
    \]
    Every vertex in $U_{2}^+$ is adjacent to $U_{1}^-\cup U_{3}^-\cup U_{4}^-$ by Proposition~\ref{prop:fork-free-results}(i, ii), and nonadjacent to $U_{2}^-$ by definition.
    Thus, all vertices in $U_{2}^+$ have the same neighborhood in $V(G)\setminus U_{2}^+$.
    Let $G_2$ be a graph of the pattern in Figure~\ref{fig:untangled}(b), where for $i = 1, 2, 3, 4$, the optional vertices $u_{i}^-$ exists if and only if $U_{i}^- \ne \emptyset$.
    It is easy to verify that $G_2$ is strongly t-perfect if it satisfies the condition of Lemma~\ref{lem:main}(ii).  We duplicate $u_{2}^+$ of $G_2$ with $|U_{2}^{+}|$ vertices.  The result is $G$.  Therefore, $G$ is strongly t-perfect by Lemma~\ref{lem:substitution}.
\end{proof}

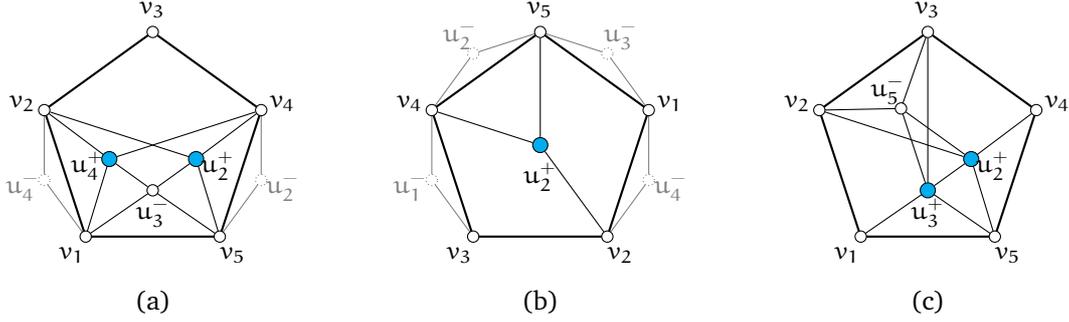
\begin{figure}[h]
  \centering
  \begin{subfigure}[b]{0.3\linewidth}
    \centering\small
    \begin{tikzpicture}[scale=1.5]
      \foreach \i in {0,..., 4} {
        \draw[thick] ({90 - (\i + 1) * (360 / 5)}:1) -- ({90 - \i * (360 / 5)}:1);
      }
      \foreach \i in {3} {
        \draw ({90 - (\i) * (360 / 5)}:1) -- ({126 - \i * (360 / 5)}:.4) node[empty vertex] (x\i) {} -- ({162 - (\i) * (360 / 5)}:1);
        \node at ({126 - \i * (360 / 5)}:.6) {$u_{\i}^-$};
      }
      \foreach \i in {2, 4} {
        \draw ({90 - (\i) * (360 / 5)}:1) -- ({126 - \i * (360 / 5)}:.4) node[filled vertex] (x\i) {} -- ({162 - (\i) * (360 / 5)}:1);
        \node at ({126 - \i * (360 / 5)}:.6) {$u_{\i}^+$};

        \draw[gray] ({90 - (\i) * (360 / 5)}:1) -- ({126 - \i * (360 / 5)}:1) node[optional vertex] (y\i) {} -- ({162 - (\i) * (360 / 5)}:1);
        \node[gray] at ({126 - \i * (360 / 5)}:1.2) {$u_{\i}^-$};
      }

      \foreach \i in {1,..., 5} {
        \node[empty vertex] (v\i) at ({306 - \i * (360 / 5)}:1) {};
        \node at ({306 - \i * (360 / 5)}:1.2) {$v_{\i}$};
      }
      \foreach \i in {2, 4} \draw (v\i) -- (x\i) -- (x3);
      \end{tikzpicture}
    \caption{}
  \end{subfigure}
  \;
  \begin{subfigure}[b]{0.3\linewidth}
    \centering\small
    \begin{tikzpicture}[scale=1.5]
      \foreach \i in {0,..., 4} {
        \draw[thick] ({90 - (\i + 1) * (360 / 5)}:1) -- ({90 - \i * (360 / 5)}:1);
      }
      \foreach \i in {1, 2, 3, 4} {
        \draw[gray] ({234 - (\i) * (360 / 5)}:1) -- ({270 - \i * (360 / 5)}:1.) node[optional vertex] (x\i) {} -- ({306 - (\i) * (360 / 5)}:1);
        \node[gray] at ({270 - \i * (360 / 5)}:1.2) {$u_{\i}^-$};
      }
      \foreach \i in {1,..., 5} {
        \node[empty vertex] (v\i) at ({90 - \i * (360 / 5)}:1) {};
        \node at ({90 - \i * (360 / 5)}:1.2) {$v_{\i}$};
      }

      \node[filled vertex, "$u_{2}^+$" below] (u) at (0, 0) {};
      \foreach \i in {2, 4, 5} \draw (u) -- (v\i);
    \end{tikzpicture}
    \caption{}
  \end{subfigure}
  \;
  \begin{subfigure}[b]{0.3\linewidth}
    \centering\small
    \begin{tikzpicture}[scale=1.5]
      \foreach \i in {0,..., 4} {
        \draw[thick] ({90 - (\i + 1) * (360 / 5)}:1) -- ({90 - \i * (360 / 5)}:1);
      }
      \foreach \i in {5} {
        \draw ({90 - (\i) * (360 / 5)}:1) -- ({126 - \i * (360 / 5)}:.4) node[empty vertex] (x\i) {} -- ({162 - (\i) * (360 / 5)}:1);
        \node at ({126 - \i * (360 / 5)}:.6) {$u_{\i}^-$};
      }
      \foreach \i in {2, 3} {
        \draw ({90 - (\i) * (360 / 5)}:1) -- ({126 - \i * (360 / 5)}:.4) node[filled vertex] (x\i) {} -- ({162 - (\i) * (360 / 5)}:1);
        \node at ({126 - \i * (360 / 5)}:.6) {$u_{\i}^+$};
       }
       \draw (x2) -- (x3);
      \foreach \i in {1,..., 5} {
        \node[empty vertex] (v\i) at ({306 - \i * (360 / 5)}:1) {};
        \node at ({306 - \i * (360 / 5)}:1.2) {$v_{\i}$};
      }
      \foreach \i in {2, 3} \draw (v\i) -- (x\i) -- (x5);
      \end{tikzpicture}
    \caption{}
  \end{subfigure}
  \caption{Three configurations for Lemma~\ref{lem:untangled} and \ref{lem:fix-5-hole}.  The dotted vertices are optional, and their edges, except to $H$, are not drawn.}
  \label{fig:untangled}
\end{figure}

Henceforth, we may assume without loss of generality that
\[
  U_{2}^{+} \ne \emptyset \text{ and } U_{3}^{+} \ne \emptyset.
\]
By Proposition~\ref{prop:fork-free-results}(\ref{lb:fork-free-v}) with $i=2$, at least one of $U_4$ and $U_5$ is empty.
For the same reason, at least one of $U_1$ and $U_5$ is empty.
We note that if neither $U_{1}\cup U_{5}$ nor $U_{4} \cup U_{5}$ is empty, then $U_{2}^{+}$ is complete to $U_{3}^{+}$, and the situation is similar to Lemma~\ref{lem:untangled}.

\begin{lemma}
  \label{lem:fix-5-hole}
  If neither $U_{1}\cup U_{5}$ nor $U_{4} \cup U_{5}$ is empty, then $G$ is strongly t-perfect.
\end{lemma}
\begin{proof}
  We first argue that $U_{5}\ne \emptyset$.  
  Suppose otherwise, then neither $U_{1}$ nor $U_{4}$ is empty. Since both $U_{3}$ and $U_{4}$ are nonempty, $U_{1}^{+}$ is empty by Proposition~\ref{prop:fork-free-results}(\ref{lb:fork-free-v}) with $i = 1$. By symmetric, $U_{4}^{+}$ is empty. Therefore $U_{1} = U_{1}^{-}$ and $U_{4} = U_{4}^{-}$. Let $u_{1}^{-}$ and $u_{4}^{-}$ be the only vertex in $U_{1}^{-}$ and $U_{4}^{-}$, respectively.
  By Proposition~\ref{prop:fork-free-results}(\ref{lb:fork-free-ii}) with $i = 3$, the vertex $u_{3}^{+}$ is adjacent to $u_{4}^{-}$.
  By Proposition~\ref{prop:fork-free-results}(\ref{lb:fork-free-i}) with $i = 3$, the vertex $u_{3}^{+}$ is adjacent to $u_{1}^{-}$.
  Since $\{u_{3}^{+}, u_{1}^{-},v_{5},u_{4}^{-},v_{2}\}$ cannot form a fork, $u_{4}^{-}$ must be adjacent to $u_{1}^{-}$.
  But then $u_{1}^{-}u_{4}^{-}v_{1}v_{5}v_{4}$ is a five-hole in $G$ and $u_{3}^{+}$ has three consecutive neighbors $u_{4}^{-}$, $v_{1}$, and $v_{5}$ on it, contradicting \myref{}.

  Since neither $U_{2}$ nor $U_{3}$ is empty, $U_{5}^{+}$ is empty by Proposition~\ref{prop:fork-free-results}(\ref{lb:fork-free-v}) with $i = 5$.  Thus, $U_{5}^{-}$ is nonempty; let $u_{5}^{-}$ be its only vertex.
  Applying Proposition~\ref{prop:fork-free-results}(\ref{lb:fork-free-i}) twice, with $i = 2,3$, respectively, we can conclude that $u_{5}^{-}$ is adjacent to all the vertices in $U_{2}\cup U_{3}$.  
  We then argue that $U_{2}^{+}$ is complete to $U_{3}^{+}$.
  We take an arbitrary vertex $u_{2}^{+}$ from $U_{2}^{+}$ and an arbitrary vertex $u_{3}^{+}$ from $U_{3}^{+}$.
  If $u_{2}^{+}$ is not adjacent to $u_{3}^{+}$, then $u_{2}^{+} v_{2} v_{3} u_{3}^{+} v_{5}$ is a hole in $G$ on which $u_{5}^{-}$ has four neighbors.
  Thus, $U_{2}^{+}$ is complete to $U_{3}^{+}$. 
  We next argue that $U_{2}^{-}$ is empty.  Suppose otherwise and let $u_{2}^{-}$ be the only vertex in $U_{2}^{-}$.
  Note that $u_{5}^{-}$ is adjacent to $u_{2}^{-}$.
  Therefore,  $u_{5}^{-}u_{2}^{-}v_{5}v_{1}v_{2}$ is a hole in $G$.
  But then, $u_{3}^{+}$ has three consecutive neighbors $u_{2}^{-}$, $v_{5}$, and $v_{1}$ on the hole, contradicting \myref{}.
  Thus, $U_{2}^{-} = \emptyset$.  A symmetric argument implies $U_{3}^{-}$ is empty as well.
  Since both $U_{2}^{+}$ and $U_{5}^{-}$ are nonempty, $U_{4}$ is empty by Proposition~\ref{prop:fork-free-results}(\ref{lb:fork-free-v}) with $i = 2$.
  Therefore, $V(G) \setminus V(H) = U_{2}^{+}  \cup U_{3}^{+} \cup U_{5}^{-}$, and the three parts $U_{2}^{+}$, $U_{3}^{+}$, and $U_{5}^{-}$ are pairwise complete with each other. 
  Let $G_1$ be the graph in Figure~\ref{fig:untangled}(c).  It is easy to verify that $G_1$ is strongly t-perfect.  We duplicate $u_{2}^+$ of $G_1$ with $|U_{2}^{+}|$ vertices, and then duplicate $u_{3}^{+}$ in the resulted graph with $|U_{3}^{+}|$ vertices.  The final result is $G$.  Therefore, $G$ is strongly t-perfect by Lemma~\ref{lem:substitution}.
\end{proof}

In the rest, at least one of $U_{1}\cup U_{5}$ and $U_{4} \cup U_{5}$ is empty.
We may assume that $U_{1}\cup U_{5} = \emptyset$; otherwise, we can renumber the vertices on $H$. 
We have seen all the maximal independent sets that contains two vertices from $H$ in Proposition~\ref{prop:independent-set-2}.
The following lists other maximal independent sets under this condition.

\begin{proposition}
  \label{prop:refined-independent-sets}
  If $U_{1} \cup U_{5}$ is empty, then a maximal independent set $S$ of $G$ that contains at most one vertex from $H$ is either
  \begin{enumerate}[i)]
  \item $U_{j}^{-} \cup \{v_{j}\}$ for some $j = 2, 3, 4$; or
  \item a pair of nonadjacent vertices $x\in U_{3}^+$ and $y\in U_{2}^+\cup U_{4}^+$.
  \end{enumerate}
\end{proposition}
\begin{proof}
  i)  Suppose first that there is one vertex in $S\cap V(H)$.  We first excludes $v_1$ and $v_5$.
  Suppose that $v_{1} \in S$.  By definition, $U_{3} \cup U_{4}$ is disjoint from $S$.  Thus, $S\subseteq U_{2} \cup \{v_{1}\}$ and cannot be maximal.  Likewise, $v_{5}\in S$ implies $S\subseteq U_{4} \cup \{v_{5}\}$.
  \begin{itemize}
  \item Case 1, $v_{2} \in S$.  Then $S \setminus \{v_{2}\}\subseteq U_{2}^{-} \cup U_{3}$.  Since $U_{3}\cup \{v_{2}, v_{4}\}$ is an independent set, $U_{2}^{-}$ cannot be empty, and its only vertex must be in $S$.  It remains to argue that the vertex in $U_{2}^{-}$ is adjacent to all the vertices in $U_{3}$.  We call Proposition~\ref{prop:fork-free-results}(\ref{lb:fork-free-ii}) with $i = 3$.
  \item Case 2, $v_{3} \in S$.  Then $S \setminus \{v_{3}\}\subseteq U_{2} \cup U_{3}^{-} \cup U_{4}$.
    Since $U_{2}$ is complete to $U_{4}$ by Proposition~\ref{prop:fork-free-results}(\ref{lb:fork-free-i}) with $i = 2$, one of $S\cap U_{2}$ and $S\cap U_{4}$ is empty.  Since $U_2\cup \{v_1, v_3\}$ and $U_4\cup \{v_3, v_5\}$ are independent sets, by the maximality of $S$, there is a vertex in $U_{3}^{-}\cap S$.      
        By Proposition~\ref{prop:fork-free-results}(\ref{lb:fork-free-ii}) with $i = 2$, the vertex in $U_{3}^{-}$ is adjacent to all the vertices in  $U_{2}$.
        Moreover, the vertex in $U_{3}^{-}$ is adjacent to all the vertices in  $U_{4}$ by Proposition~\ref{prop:fork-free-results}(\ref{lb:fork-free-iii}) with $i = 2$ when $U_{4}^{+}=\emptyset$, or by Proposition~\ref{prop:fork-free-results}(\ref{lb:fork-free-ii}) with $i = 4$ otherwise.

  \item Case 3, $v_{4} \in S$.
    Then $S \setminus \{v_{4}\}\subseteq U_{4}^{-} \cup U_{3}$, and the argument is similar to that of case 1.
  \end{itemize}
  
  ii) Now suppose that $S$ is disjoint from $V(H)$.  By assumption, $V(G)\setminus V(H) = U_{2} \cup U_{3} \cup U_{4}$.
  We first argue that
  \[
    S\subseteq U_{2}^+ \cup U_{3}^+ \cup U_{4}^+.
  \]
  For $j = 2, 3, 4$, let $x_j$ be the vertex in $U_{j}^-$ if $U_{j}\ne U_{j}^+$.  Applying Proposition~\ref{prop:fork-free-results}(\ref{lb:fork-free-i}, \ref{lb:fork-free-iii}) with $i=2$, we can conclude that $x_{2}$, $x_{3}$, and $x_{4}$ are pairwise adjacent, when they exist.  Therefore, at most one of them is in $S$.
  On the other hand, if $x_j \in S$ for $j = 2, 3, 4$, then $S\subseteq U_j$ by Proposition~\ref{prop:fork-free-results}(\ref{lb:fork-free-i}, \ref{lb:fork-free-ii}).
  Since this contradicts the maximality of $S$, we must have $S\subseteq U_{2}^+ \cup U_{3}^+ \cup U_{4}^+$.
  Since $U_{2}^{+}$ is complete to $U_{4}^{+}$ by Proposition~\ref{prop:fork-free-results}(\ref{lb:fork-free-i}) with $i=2$, the set $S$ is a subset of either $U_{2}^{+} \cup U_{3}^{+}$ or $U_{3}^{+} \cup U_{4}^{+}$.  
  By Proposition~\ref{prop:fork-free-results}(\ref{lb:fork-free-iv}) (with $i = 3$), each vertex in $U_{3}^{+}$ has at most one non-neighbor in $U_{2}^{+}$ and at most one non-neighbor in $U_{4}^{+}$.  For the same reason, each vertex in $U_{2}^{+}$ or $U_{4}^{+}$ has at most one non-neighbor in $U_{3}^{+}$.  Thus, $S$ is a pair of nonadjacent vertices $x\in U_{3}^+$ and $y\in U_{2}^+\cup U_{4}^+$.
\end{proof}

The final step of the proof relies on the duality of linear programming, for which we need to recall some known results.
For each weight function $w:V(G) \rightarrow \mathbb{Z}_{\geq 0}$, we can make a linear program out of \eqref{eq:t-perfect} by adding an objective function $\max \sum_v w(v) x_v$.  The dual of this linear program is a covering problem.
A \emph{$w$-cover} is a family of vertices, edges, and odd cycles in $G$ such that every vertex $v$ in $V(G)$ lies in at least $w(v)$ elements, with repetition allowed.
The \emph{cost} of a $w$-cover is the sum of the costs of its elements, where the cost of a vertex or an edge is one, and the cost of an odd cycle $C$ is $(|C|-1)/2$.
For a vertex set $S$, we use $w(S)$ to denote $\sum_{v \in S} w(v)$.  We use $\alpha_{w}(G)$ to denote the maximum value of $w(S)$, with $S$ ranging over all independent sets of $G$.
The following is a consequence of linear programming duality.

\begin{proposition}[{\cite{schrijver-03}}]
  \label{prop:strong-t-perfection}
  A graph $G$ is \emph{strongly} t-perfect if and only if there exists a $w$-cover of cost $\alpha_{w}(G)$ for every weight function $w:V(G) \rightarrow \mathbb{Z}_{\geq 0}$.
\end{proposition}

The following observation of Bruhn and Stein~\cite{Bruhn10} is very helpful in our checking the condition of Proposition~\ref{prop:strong-t-perfection}.  We provide a proof for the sake of completeness.
Note that a vertex set $K$ intersects every maximum-weight independent set of $G$ if and only if $\alpha_w(G - K) < \alpha_w(G)$.
\begin{proposition}[\cite{Bruhn10}]
  \label{proposition:clique-intersects-mwis}
  Let $G$ be a graph and $w: V(G) \rightarrow \mathbb{Z}_{\geq 0}$ a weight function.   There exists a $w$-cover of $G$ with cost $\alpha_{w}(G)$ if
  \begin{itemize}
  \item there exists a clique $K$ of at most three vertices such that $\alpha_w(G - K) < \alpha_w(G)$; and
  \item for any weight function $w': V(G) \to \mathbb{Z}_{\geq 0}$ such that $w'(V(G)) < w(V(G))$, there exists a $w'$-cover of cost $\alpha_{w'}(G)$.
  \end{itemize}
\end{proposition}
\begin{proof}
  We may assume without loss of generality that $K$ is inclusion-wise minimal satisfying $\alpha_w(G - K) < \alpha_w(G)$.
  As a result, $w(v) > 0$ for each $v\in K$: a vertex of zero weight has no impact on $\alpha_w(G)$.  We can define another weight function $w': V(G)\to \mathbb{Z}_{\geq 0}$ by setting

  \begin{equation*}
    w'(v) =
    \begin{cases}
      w(v)-1 & v \in K, \\
      w(v) & \text{otherwise.}
    \end{cases}
  \end{equation*}
  Since $w'(V(G)) < w(V(G))$, there exists a $w'$-cover $\mathcal{K}$ of cost $\alpha_{w'}(G)$ by assumption.
  Since $|K| \le 3$, the set $\mathcal{K} \cup \{K\}$ is a $w$-cover of $G$ and its cost is $\alpha_{w'}(G) + 1 = \alpha_{w}(G)$.
\end{proof}

\begin{lemma}
  \label{lem:sufficient-condition-for-t-perfection}
  If $ U_{1} \cup U_{5}$ is empty, then $G$ is strongly t-perfect.
\end{lemma}
\begin{proof}
  Suppose for contradiction that $G$ is not strongly t-perfect, and assume without loss of generality that $G$ is a counterexample of the minimum number of vertices.
  Our first claim is that every proper induced subgraph $G'$ of $G$ is strongly t-perfect.  If $G'$ does not contain an odd hole, then it is strongly t-perfect (Proposition~\ref{prop:perfect and t-perfect}).
  If $G'$ is claw-free, then $G$ is strongly t-perfect~\cite{Bruhn10, Bruhn12}.
  Thus, $G'$ is strongly t-perfect either because it satisfies one of Lemmas~\ref{lem:untangled} and \ref{lem:fix-5-hole}, or by the selection of $G$.

  By Proposition~\ref{prop:strong-t-perfection}, there exists a weight function $w:V(G) \rightarrow \mathbb{Z}_{ \geq 0}$ such that $G$ does not have a $w$-cover of cost $\alpha_{w}(G)$.  We may take $w$ to be such a function that minimizes $w(V(G))$.
  The second claim is that the weight is positive.  Suppose that $w(v) = 0$ for some vertex $v \in V(G)$.
  Since every induced subgraph of $G$ is strongly t-perfect, there exists a $w$-cover $\mathcal{K}$ of $G-v$ with cost $\alpha_{w}(G-v)$.  Since $w(v) = 0$, the cover $\mathcal{K}$ is also a $w$-cover of $G$, while $\alpha_{w}(G-v) = \alpha_{w}(G)$.  But then $\mathcal{K}$ is a $w$-cover of $G$ with cost $\alpha_{w}(G)$, a contradiction.
  As a consequence of the second claim, every maximum-weight independent set is maximal.  Recall that all maximal independent sets are listed in Propositions~\ref{prop:independent-set-2} and~\ref{prop:refined-independent-sets}.

  For $j=2,3,4$, let
  \[
    S^-_j = \{v_{j-1}, v_{j+1}\}\cup U^-_j
  \]
  and denote by $u_{j}^{-}$ the only vertex contained in $U_{j}^{-}$ when it is not empty.
  For $j=2,3$, let $u^+_j$ be a vertex of the maximum weight from $U^+_j$, and 
  \[
    S^+_j = \{v_{j-1}, v_{j+1}, u^+_j\}.
  \]
  We define a set $S^+_4 = \{v_{3}, v_{5}, u^+_4\}$ when $u^+_4 \ne \emptyset$, with $u^+_4$ being a vertex of the maximum weight from $U^+_4$. 
  According to Proposition~\ref{prop:independent-set-2}, all the nine sets $S^-_j$, $S^+_j$, and $U_j$ are independent sets.

  From Proposition~\ref{proposition:clique-intersects-mwis} and the selection of the weight function $w$ it can be inferred that $\alpha_w(G - K) = \alpha_w(G)$ for any clique $K$ of at most three vertices.  In other words, there exists a maximum-weight independent set $S$ of $G$ disjoint from $K$.
  We try to locate a clique of two or three vertices that intersects all maximum-weight independent sets of the graph, thereby producing a contradiction to Proposition~\ref{proposition:clique-intersects-mwis}.
In the following we consider potential maximum-weight independent sets.  By excluding an independent set we mean that we have evidence that it does not have the maximum weight.

  Note that $U_{4}$ is not empty; otherwise, every odd cycle of $G$ visits $v_{5}$, and $G$ is strongly t-perfect according to Gerards~\cite{Gerards89}.
  We take an arbitrary vertex $u_4$ from $U_{4}$.  
  It is adjacent to $u_2^+$ by Proposition~\ref{prop:fork-free-results}(\ref{lb:fork-free-i}) with $i=2$.
  Let $K$ denote the clique $\{v_{2}, u_2^+, u_4\}$, and let $S$ be a maximum-weight independent set of $G$ disjoint from $K$.
  Note that $S$ has to be $\{v_{1}, v_{4}\}$, $\{v_{3}\} \cup U_{3}^{-}$, $\{v_{4}\} \cup U_{4}^{-}$, or one that is disjoint from $V(H)$, i.e., specified in Proposition~\ref{prop:refined-independent-sets}(ii).
  
  \begin{itemize}
  \item Case 1, $S = \{v_{1}, v_{4}\}$.  (Note that $U_5 = \emptyset$.)
    Since $\{v_2, v_4, u^+_3\}$ and $\{v_1, v_3, u^+_2\}$ are both independent sets,
    
    \begin{align*}
      \label{eq:1}
      &w(u^+_2) + w(u^+_3)
      \\ < &
             w(v_2) + w(v_4) + w(u^+_3) + w(v_1) + w(v_3) + w(u^+_2) - w(v_4) - w(v_1)
      \\ = &
             w(\{v_2, v_4, u^+_3\}) + w(\{v_1, v_3, u^+_2\}) - w(S)
      \\ \leq & \alpha_w(G) + \alpha_w(G) - \alpha_w(G)
      \\ = & \alpha_w(G).
    \end{align*}

    By the selection of $u^+_2$ and $u^+_3$, a pair of vertices $x\in U_{2}^+$ and $y\in U_{3}^+$ cannot have weight $\alpha_w(G)$.  In other words, if a maximum-weight independent set is disjoint from $V(H)$, then it comprises a vertex in $U_{3}^+$ and a vertex in $U_{4}^+$.  On the other hand, from
    \[
      w(v_2) + w(v_3) + w(U^-_2\cup U^-_3) = w(S^-_2) + w(S^-_3) - w(S) \le \alpha_w(G)
    \]
    we can exclude $\{v_{2}\} \cup U_{2}^{-}$ and $\{v_{3}\} \cup U_{3}^{-}$.
    Thus, if a maximum-weight independent set contains one vertex from $H$, then it has to be $\{v_{4}\} \cup U_{4}^{-}$.
    \begin{itemize}
    \item Case 1.1,  $\{v_{2}, v_{5}\}$ is also a maximum-weight independent set.  (Note that $U_1 = \emptyset$.) 
      If $U_{4}^{+}\ne \emptyset$, we use
      \[
        w(u^+_3) + w(u^+_4) < w(S^+_3) + w(S^+_4) - w(\{v_{2}, v_{5}\}) \le \alpha_w(G)
      \]
      to exclude all maximal independent sets disjoint from $H$.
      If $U_{4}^{-}\ne \emptyset$, we use
      \[
        w(v_4) + w(u_{4}^{-}) < w(S^-_3) + w(S^-_4) - w(\{v_{2}, v_{5}\}) < \alpha_w(G)
      \]
      to exclude $\{v_{4}, u_{4}^{-}\}$.  Since any maximum-weight independent set has to contain two vertices from $H$, they all intersect the clique $\{v_{1}, v_{5}, u_{3}^{+}\}$.
    \item Case 1.2, there exists a maximum-weight independent set $S' = \{x_{3}, x_{4}\}$ with $x_{3}\in U_{3}^+$ and $x_{4}\in U_{4}^+$.
      Note that both $U_{3} \cup \{v_{2}\}$ and $U_{4} \cup \{v_{3}\}$ are not maximal.
      Therefore, we can use $w(U_{3} \cup \{v_{2}\}) + w(U_4 \cup \{v_{3}\}) - w(S') < \alpha_w(G)$ to exclude all other pairs $\{x'_3, x'_4\}$ with $x'_{3}\in U_{3}^+$ and $x'_{4}\in U_{4}^+$ (except for $S'$ itself).
      If $U_{4}^{-}$ is empty, then $\{v_{1},v_{2},x_{4}\}$ intersects all the possible maximum-weight independent sets.
      Now that $U_{4}^{-}$ is nonempty, we use $w(U_{4}) + w(S^+_3) - w(S') < \alpha_w(G)$ to exclude $\{v_{4}, u_{4}^{-}\}$.
      Furthermore, we use $w(S^+_3) + w(S^+_4) - w(S') < \alpha_w(G)$ to exclude $\{v_{2}, v_{5}\}$.
 The clique $\{v_{1}, v_{5}, u_{3}^{+}\}$ intersects all the remaining maximal independent sets, $S$, $S'$, $\{v_{2},v_{4}\} \cup U_{3}$, $\{v_{3},v_{5}\} \cup U_{4}$, and $\{v_{1},v_{3}\} \cup U_{2}$.
    \item Otherwise (neither of cases 1.1 and 1.2 is true),
      the clique $\{v_{3}, v_{4}\}$ intersects all the possible maximum-weight independent sets.
    \end{itemize}
    
  \item Case 2, $S = \{v_{3}, u_{3}^{-}\}$.
    We use
    \[
      w(u_{2}^{+}) + w(u_{3}^{+}) < w(U_{3}) + w(S^+_2) - w(S) < \alpha_w(G)
    \]
    to exclude all pairs $\{x_2, x_3\}$ with $x_{2}\in U_{2}^+$ and $x_{3}\in U_{3}^+$.
    If $U_{4}^{+}$ is nonempty, we use
    \[
      w(u_{3}^{+}) + w(u_{4}^{+}) < w(U_3) + w(S^+_4) - w(S) < \alpha_w(G)
    \]
    to exclude all pairs $\{x_3, x_4\}$ with $x_{3}\in U_{3}^+$ and $x_{4}\in U_{4}^+$.
    From $w(S^-_3) + w(S^-_4) - w(S) < \alpha_w(G)$ we can exclude $\{v_{2}, v_{5}\}$ and $\{v_{4}, u_{4}^{-}\}$ (when $U^-_4 \ne \emptyset$).
    If $U^-_2 \ne \emptyset$, we use  $w(S^-_2) + w(S^-_3) - w(S) < \alpha_w(G)$ to exclude $\{v_{2}, u_{2}^{-}\}$.   We are left with $S$, $\{v_{2},v_{4}\} \cup U_{3}$, $\{v_{3},v_{5}\} \cup U_{4}$, and $\{v_{3},v_{1}\} \cup U_{2}$.   All of them intersect the clique $\{v_{3}, v_{4}\}$.
    
  \item Case 3, $S = \{v_{4}, u_{4}^{-}\}$.
    We can use
    \[
      w(v_{2}) + w(v_{5}) < w(S^-_3) + w(S^-_4) - w(S) \le \alpha_w(G)
    \]
    to exclude $\{v_{2}, v_{5}\}$.
    If $U_{4}^{+}\ne \emptyset$, we use
    \[
      w(u_{3}^{+}) + w(u_{4}^{+}) < w(U_{4}) + w(S^+_3) - w(S) < \alpha_w(G)
    \]
    to exclude all pairs $\{x_3, x_4\}$ with $x_{3}\in U_{3}^+$ and $x_{4}\in U_{4}^+$.
    \begin{itemize}
    \item Case 3.1, there does not exist a maximum-weight independent set $\{x_2, x_3\}$ with $x_{2}\in U_{2}^+$ and $x_{3}\in U_{3}^+$.
      If $U_{2}^{-}$ is empty, the clique $\{v_{3},v_{4}\}$ intersects all maximum weight independent sets.
      Now that $U_{2}^{-}\ne \emptyset$, we note that $\{u_{2}^{-}, u_{3}^{+}, u_{4}^{-}\}$ intersects all maximum-weight independent sets.  To see that it is clique, note that $u_{2}^{-}$ is adjacent to $u_{4}^{-}$ by Proposition~\ref{prop:fork-free-results}(\ref{lb:fork-free-i}) with $i=2$, and $u_{3}^{+}$ is adjacent to both $u_{2}^{-}$ and $u_{4}^{-}$ by Proposition~\ref{prop:fork-free-results}(\ref{lb:fork-free-ii}) with $i=3$,
    \item Case 3.2, there exists a maximum-weight independent set $S' = \{x_2, x_3\}$ with $x_{2}\in U_{2}^+$ and $x_{3}\in U_{3}^+$.
 We can use $w(U_{2} \cup \{v_{1}\}) + w(U_3 \cup \{v_{2}\}) - w(S') < \alpha_w(G)$ to exclude all other pairs $\{x'_2, x'_3\}$ with $x'_{2}\in U_{2}^+$ and $x'_{3}\in U_{3}^+$ (except for $S'$ itself).
 If $U_{2}^{-}$ is not empty, then we can use $w(U_{2}) + w(S^+_3) - w(S') < \alpha_w(G)$ to exclude $\{u_{2}^{-},v_{2}\}$.  Thus, a maximum-weight independent set of $G$ has to be $S$, $S'$, $\{v_{2},v_{4}\} \cup U_{3}$, $\{v_{3},v_{5}\} \cup U_{4}$, or $\{v_{3},v_{1}\} \cup U_{2}$.  The clique $\{v_{4}, v_{5}, x_{2}\}$ intersects all maximum-weight independent sets.
    \end{itemize}    
    
  \item Case 4, $S = \{x_{2}, x_{3}\}$, where $x_{2}\in U_{2}^+$ and $x_{3}\in U_{3}^+$.  Note that $x_{2} \neq u_{2}^+$ because $S$ is disjoint from $K$.
 We can use $w(U_{2} \cup \{v_{1}\}) + w(U_3 \cup \{v_{2}\}) - w(S) < \alpha_w(G)$ to exclude all other pairs $\{x'_2, x'_3\}$ with $x'_{2}\in U_{2}^+$ and $x'_{3}\in U_{3}^+$ (except for $S$ itself).
 If $U_{2}^{-}$ is nonempty, then we can use $w(U_{2}) + w(S^+_3) - w(S) < \alpha_w(G)$ to exclude $\{u_{2}^{-},v_{2}\}$.
 If no maximum-weight independent set intersects $U_{4}^+$, then the clique $\{x_{2},v_{4},v_{5}\}$ intersects all maximum weight independent sets.

 Suppose that $S' = \{x'_{3}, x_{4}\}$ is a maximum-weight independent set with $x'_{3} \in x_{3}$ and $x_{4} \in U_{4}^{+}$.  We can further use $w(U_{3} \cup \{v_{4}\}) + w(U_4 \cup \{v_{3}\}) - w(S') < \alpha_w(G)$ to exclude all other pairs $\{x''_3, x'_4\}$ with $x''_{3}\in U_{3}^+$ and $x'_{4}\in U_{4}^+$ (except for $S'$ itself).
 We use $w(S^+_3) + w(S^+_4) - w(S') \le \alpha_w(G)$ to exclude $\{v_{2},v_{5}\}$.  Thus, a maximum-weight independent set of $G$ has to be $S$, $S'$, $\{v_{2},v_{4}\} \cup U_{3}$, $\{v_{3},v_{5}\} \cup U_{4}$, or $\{v_{3},v_{1}\} \cup U_{2}$.  The set $\{v_{4},x_{2},x_{4}\}$ intersects all maximum-weight independent sets.  Note that $x_{2}$ is adjacent to $x_{4}$ by Proposition~\ref{prop:fork-free-results}(\ref{lb:fork-free-i}) with $i=2$.
    
  \item Case 5, $S = \{x_{3}, x_{4}\}$, where $x_{3}\in U_{3}^+$ and $x_{4}\in U_{4}^+$.  It is similar to Case 4.
  \end{itemize}
  This concludes this proof.
\end{proof}

Lemma~\ref{lem:main} follows from Lemmas~\ref{lem:untangled}, \ref{lem:fix-5-hole}, and \ref{lem:sufficient-condition-for-t-perfection}.

\section{Recognition}

We now describe an algorithm to decide whether a fork-free graph is (strongly) t-perfect.
We may assume without loss of generality that the input graph is connected; otherwise, we work on its components one by one.
The algorithm is based on Lemma~\ref{lem:main}.  The only condition of Lemma~\ref{lem:main}(ii) that cannot be easily checked in polynomial time is that every odd hole has length five.
The following proposition bounds the length of the longest odd holes.

\begin{proposition}
  \label{lem:bound}
  Let $G$ be a $\{K_{4}, \mathrm{fork}\}$-free graph containing a five-hole $H$.  If $H$ satisfies \myref{}, then $G$ cannot contain an odd hole with length longer than $19$.
\end{proposition}
\begin{proof}
  Let $H'$ be a longest odd hole in $G$.  Suppose for contradiction $|H'| \ge 21$.  At least $|H'|-4$ vertices of $H'$ are in $V(G) \setminus V(H)$.
  Assume without loss of generality that $|U_{i} \cap V(H')|$ is maximized with $i = 1$.
  Then $|U_{1} \cap V(H')| \geq \lceil \frac{|H'|-4}{5} \rceil \geq 4$.
  Since $U_{1}$ is an independent set, $|U_{1} \cap V(H')| \leq \frac{|H'|-1}{2}$.  There exists a vertex $x$ in $V(H') \setminus (V(H)\cup U_{1})$.  By Propositions~\ref{prop:fork-free-results}(\ref{lb:fork-free-i}), (\ref{lb:fork-free-ii}), and (\ref{lb:fork-free-iv}) with $i=1$, the vertex $x$ has at most one non-neighbor in $U_{1} \cap V(H')$.  But then $x$ has at least three neighbors in $|H'|$, contradicting that $H'$ is a hole.
\end{proof}

We are now ready to present the recognition algorithm and prove Theorem~\ref{thm:fork-free-recognize}.

\begin{proof}[Proof of Theorem~\ref{thm:fork-free-recognize}]
  The input is a fork-free graph $G$.
  We start by checking whether it contains a $K_4$, $W_{5}$, $C^2_7$, or $C_{10}^{2}$.
  Since $K_4$, $W_{5}$, $C^2_7$, and $C_{10}^{2}$ are not t-perfect, we return  ``no'' if any of them is found.
  If $G$ does not contain a claw, then we call Bruhn--Schaudt algorithm~\cite{Bruhn16} to decide whether $G$ is t-perfect.
  We then call the algorithm of Chudnovsky et al.~\cite{Chudnovsky20} to test whether $G$ contains an odd hole.
  Since $G$ does not contain $K_4$, it cannot contain the complement of any odd hole longer than seven.  It does not contain $C^2_7$, which is the complement of $C_7$.  Therefore, if $G$ does not contain any odd hole, then $G$ is perfect, and t-perfect (Proposition~\ref{prop:perfect and t-perfect}), and we return ``yes.''
  In the rest, $G$ contains a claw and an odd hole, and   we check the conditions of Lemma~\ref{lem:main}(ii).
  We enumerate five-holes, and for each of them, test whether it satisfies \myref{}.  If any one does not, then return ``no.''
  Finally, we check whether $G$ contains an odd hole of length between $7$ and $19$.  If any is found, then we  can return ``no.''  If none is found, every odd hole has length five by Proposition~\ref{lem:bound}.  Thus, we can return ``yes.''
  All the induced subgraphs we need to check have a constant number of vertices, and both algorithms we call~\cite{Bruhn16, Chudnovsky20} take polynomial time.
  Thus, the whole algorithm runs in polynomial time.
\end{proof}

\end{document}